%% file: articleAxel.tex
\documentclass[a4paper]{article}

\usepackage[utf8]{inputenc}
\usepackage[T1]{fontenc}
\usepackage{lmodern}
\usepackage[english]{babel}

\usepackage{amsmath,amssymb}
\usepackage{theorems}

\usepackage[ruled,titlenumbered,linesnumbered,boxed,noend]{algorithm2e}

\usepackage{tikz}

\newcommand\abs[1]{\ensuremath{\lvert#1\rvert}}

\newcommand\ie{\textit{i.e.}}
\newcommand\cM{\ensuremath{\mathcal M}}
\def\cA{\mathcal{A}}
\def\cB{\mathcal{B}}
\def\cC{\mathcal{C}}

\def\cM{\mathcal{M}}

\def\cU{\mathcal{U}}
\def\cZ{\mathcal{Z}}

\title{Efficient random sampling of binary and unary-binary trees \textit{via}
holonomic equations}
\author{Axel Bacher$^\dagger$\and Olivier Bodini$^\ddagger$\and Alice Jacquot$^\ddagger$}
\date\today
\begin{document}

\maketitle

{\footnotesize
\noindent $\dagger$ : Research Institute for Symbolic Computation (RISC),
Johannes Kepler University,
Altenbergerstraße 69,
A-4040 Linz, Austria. E-mail: \texttt{abacher@risc.jku.at}\\
$\ddagger$ :{Universit\'e Paris 13, Sorbonne Paris Cit\'e, LIPN, CNRS(UMR 7030),
F-93430, Villetaneuse, France. E-mail: \texttt{bodini@lipn.fr},
\texttt{jacquot@lipn.fr}.}
}

\begin{abstract}
We present a new uniform random sampler for binary trees with~$n$ internal
nodes consuming $2n + \Theta(\log(n)^2)$ random bits on average. This makes it
quasi-optimal and out-performs the classical Remy algorithm. We also present a
sampler for unary-binary trees with~$n$ nodes taking $\Theta(n)$ random bits
on average. Both are the first linear-time algorithms to be optimal up to a
constant.
\end{abstract}


\section{Introduction}

Rooted Plane trees are central data structures in computer science, and have been widely studied in both mathematics and computer science. They are a natural way of representing hierarchy and arise in a huge number of algorithms such as depth first search, sorting algorithms and search algorithms. 
In this paper, we focus on uniform random sampling of Catalan and Motzkin trees (\ie, rooted plane binary tees and rooted plane unary-binary tees). 
More precisely, we address the search of a algorithm which returns uniformly a tree having $n$ nodes with minimal cost in terms of random bit, while keeping good time and space complexities. 
Contrariwise with the uniform-real-variable model where the measure of complexity is the number of calls to iid uniform real variables $U$ on the compact $[0,1]$, we deal with the realistic and tractable random-bit model introduced by Von Neumann and further developed by Knuth and Yao. In this model, the unit of complexity is the random bit. For obvious reasons, this notion is much more natural in computer science and coherent with Shannon information entropy \cite{shannon1948mathematical}. This notion have already been investigated \cite{FlPeSo11,conf/wsc/Duchon11}
It is clear that the uniform random sampler of an object $\gamma$ inside a set of cardinality $C$ needs at least $\ln_2(C)$ random bits. 

%Recently, sampler efficient in the number of bits have be developped, machines de buffon, duchon ?

 The efficient (in time space and random bits) sampling of trees is of central interest in many domain. Since these structures appear in numerous settings (data structures, bioinformatics, probabilities,...), an efficient uniform drawing enables testing, conjecturing, and simulating important properties. 

The generation of trees is quite a long story. The classical ways of sampling trees include the recursive method \cite{FZV94}, the Boltzmann method \cite{bolt1}, bijection with words \cite{Alonso96alinear,Alonso1994529}, bijection with walks \cite{Barcucci19953} and, for Catalan trees, R\'emy's Algorithm \cite{Remy85} and Galton-Watson branching processes. 

Rémy's Algorithm is certainly the most efficient algorithm until now. It consists to grow up a tree by local transformations that we can call grafting. Our first contribution is to reinterpret this algorithm as based on a holonomic equation (linear
differential equations with polynomial coefficients) over Catalan trees. This algorithm runs in linear time and space if we consider that uniformly choosing a node over $n$ in a tree can be done in constant time, and uses $\Theta ( n \log n )$ random bits in every case. 
We improve this algorithm such that it still runs in linear time in average and only uses a nerly optimum number of random bits, in average too. 

Actually this method is not limited to Catalan trees: there is a similar, though more intricate, holonomic equation over Motzkin tree. We give a combinatorial interpretation to this equation, which gives us a way to grow a Motzkin tree by local modifications. Thus, we can extend our version of Rémy's Algorithm to Motzkin trees, in average linear time and random bits. 

We present in the first section the bijections on catalan and Motzkin trees leading to holonomic specifications, then in Section~\ref{sec:random} we give our sampler and analyse them. Finally, in Section~\ref{sec:weight} we give a bivariate specification for weighted unary-binary trees. 

\section{Holonomic specification of trees}

\subsection{Holonomic equations}

% Cette section sera à reprendre une fois l'introduction écrite, pour éviter
% les redondances.

The goal of this section is to find combinatorial specifications~\cite{livre} for the
classes of binary and Motzkin trees based on holonomic equations (linear
differential equations with polynomial coefficients) on their generating
functions rather than algebraic equations. In other words, these
specifications may involve pointing but no branching.

Finding a holonomic equation from an algebraic one can be done automatically
with a computer algebra system (for instance Maple with the \texttt{gfun}
package~\cite{SaZi94}). In the cases of binary and Motzkin trees, these
equations are:
\begin{align*}
B(z) + (z - 4z^3)B'(z) &= 2z\text;\\
(1 - z) M(z) + (z - 2z^2 - 3z^3) M'(z) &= 2z\text.
\end{align*}
For $A = B$ or $M$, set $A^\bullet(z) = zA'(z)$ and $A^\star(z) = A(z) +
A^\bullet(z)$. The above equations rewrite into:
\begin{align} \label{bdiff}
B^\star(z) &= 2z + 4z^2B^\bullet(z)\text;\\ \label{mdiff}
M^\star(z) &= 2z + zM^\star(z) + zM^\bullet(z) + 3z^2M^\bullet(z)\text.
\end{align}

Combinatorially, we denote by $\mathcal B^\bullet$ and $\mathcal M^\bullet$
the classes of binary trees and Motzkin trees, respectively, with a pointed
node. We also denote by $\mathcal B^{(\ell)}$ and $\mathcal M^{(\ell)}$ the
classes of trees pointed on a leaf and by $\mathcal M^{(u)}$ the class of
Motzkin trees pointed on a unary node. We define:
\begin{align*}
\mathcal B^\star &= \mathcal B^{(\ell)} + \mathcal B^{(\ell)}\text;\\
\mathcal M^\star &= \mathcal M^{(\ell)} + \mathcal M^{(\ell)} + \mathcal
M^{(u)}\text,
\end{align*}
In other words, the class $\mathcal B^\star$ is the class of binary trees
pointed in two possible ways on a leaf, and the class $\mathcal M^\star$ is
the class of Motzkin trees pointed either on a leaf in two possible ways, or
on a unary node. For convenience, we refer to these points with colors: a tree
of~$\mathcal B^\star$ can have either a blue pointed leaf or a red pointed
leaf, and a tree of~$\mathcal M^\star$ can have a blue pointed leaf, a red
pointed leaf, or a green pointed unary node. We refer to~$T$ as a
\emph{color pointed tree}.

Obviously, the respective generating functions of $\mathcal B^\bullet$ and
$\mathcal M^\bullet$ are $B^\bullet(z)$ and $M^\bullet(z)$. In
Section~\ref{sec:graft}, we give combinatorial interpretations of the equations
\eqref{bdiff} and \eqref{mdiff}. In Section~\ref{sec:repoint}, we complete the
specification by interpreting the equations $B^\star(z) = B(z) + B^\bullet(z)$
and $M^\star(z) = M(z) + M^\bullet(z)$.

\subsection{The graft operation} \label{sec:graft}

This section first give a bijection $F$ for Catalan trees:
\[F: 2 \mathcal{Z} + 4\mathcal{Z}^2 \cB^\bullet \rightarrow \cB^\star \]
and then the bijection $G$ for Motzkin trees:
\[G:2\cZ + \cZ \cM^\star + \cZ \cM^\bullet + 3\cZ ^2\cM^\bullet \rightarrow \cM^\star.\] 

These bijections are based on the idea of grafting, \ie , a local growth which can be locate anywhere in the tree (by opposition to most recursion which make trees grow from the root or the leaves)
The reverse bijections are given in proofs. 

\subsubsection{Bijection $F$ for Catalan trees}

The term $2 \mathcal{Z}$ correspond to the two trees reduced to a leaf, either blue pointed or red pointed. 
Now, there is $4$ ways to grow a tree from $\cB^\bullet$: \\
Let $T$ be be pointed tree of $\cB^\bullet$. Let $n$ be the pointed node, $p$ its parent (if $n$ is the root, $p=\bot$). There is now $4$ possible cases, named $F1$, $F2$, $F3$ and $F4$. In each of this case, we create a new binary node $b$ and a new leaf $\ell$ such that $b$ become the parent of $n$, $p$ become the parent of $b$ and either:
\begin{itemize}
\item in case $F1$, $\ell$ is the right child of $b$ and is red pointed
\item in case $F2$, $\ell$ is the right child of $b$ and is blue pointed
\item in case $F3$, $\ell$ is the left child of $b$ and is red pointed
\item in case $F4$, $\ell$ is the left child of $b$ and is blue pointed
\end{itemize}
and the original point on $n$ is forgotten. 

This bijection is represented on Figure~\ref{fig:bijremy}, where apointed node is represented by 
\begin{tikzpicture}[inner sep=1pt, circle, minimum size=0pt, xscale=0.75, scale =0.8]
\node [draw, minimum size=5pt] (n) at (0,-1){$\bullet$};\end{tikzpicture}, a
 blue pointed leaf by \begin{tikzpicture}[inner sep=1pt, circle, minimum size=0pt, xscale=0.75, scale =0.8]
\node [blue, draw, minimum size=4pt] (l) at (2,-1){\tiny $\nwarrow$};
\end{tikzpicture} and a red pointed leaf by \begin{tikzpicture}[inner sep=1pt, circle, minimum size=0pt, xscale=0.75, scale =0.8]
\node [red, draw, minimum size=4pt] (l) at (2,-1){\tiny $\nearrow$};
\end{tikzpicture}.

\input{bijremyaxel}

\begin{proof}
First observe that there is a unique way in $F$ to obtain each tree red pointed on a left leaf, red pointed on a right leaf, blue pointed on a right leaf and blue pointed on a left leaf. 

Now, we describe the reverse bijection, $F^{-1}$ which, from a tree (of size $3$ or more) of $\cM^\star $ give a unique tree of 
$\cB^\bullet$. Let $\ell$ be the (red or blue) pointed leaf, let $b$ be its parent (as the size of the tree is at least $3$ a leaf cannot be the root), $s$ the other child of $b$ ($s$ can be a binary node or a leaf) and $p$ the parent of $b$. We supress $\ell$ and $b$ and $p$ become the parent of $s$: if $b$ where a left (resp.right) child, $s$ become a left (resp. right) child. We place a point on $s$.

By composing $F$ and $F^{-1}$, we have $F(F^{-1}(T)) = T$ and $F^{-1}(F(T')) = T'$. 
\end{proof}

\subsubsection{Bijection $G$ for Motzkin trees}

The identity
\[2\cZ + \cZ \cM^\star + \cZ \cM^\bullet + 3\cZ ^2\cM^\bullet =\cM^\star\] 
can be translated to the case it produces Motzkin trees of size $n$, with $n\geq 3$: 
\[G_n:\cZ \cM^\star_{n-1} + \cZ \cM^\bullet_{n-1} + 3\cZ ^2\cM^\bullet_{n-2} \rightarrow \cM^\star_{n}.\]

Now we decompose the operation $G_n$ according to its starting set. These operators are represented in Figure~~\ref{fig:bijmotz}, where a green pointed unary node is represented by 
\begin{tikzpicture}[inner sep=1pt, circle, minimum size=0pt, xscale=0.75, scale =0.8]\node [green, draw, minimum size=4pt](u) at (0,0){$\downarrow$};\end{tikzpicture}.

\input{bijmotzaxel}
\begin{itemize}
\item $G1_n: \cM^\star_{n-1} \rightarrow \cM^\star_{n}$
The operation depend of the color of the point.
\begin{itemize}
\item If we start from a tree red pointed on a leaf $\ell$, we call $p$ the parent of $\ell$. The operation $G1_n$ consists in adding a new unary node as the new parent of $\ell$ and (left, right or only according to what $\ell$ was) child of $p$. We keep the red point on $\ell$. 
\item If we start from a tree blue pointed on a leaf $\ell$, we do a similar modification and keep the blue point on $\ell$. 
\item If we start from a tree green pointed on a unary node, we note $p$ its parent and $n$ its child. The operation $G1_n$ consists in transforming $u$ into a binary node of parent $p$, left child $n$ and right child a newly created leaf $\ell$. The new tree is considered red pointed on $\ell$. 
\end{itemize}
\item $G2_n: \cM^\bullet_{n-1} \rightarrow \cM^\star_{n}$
We start from a tree pointed on a node $n$, of parent $p$. The new tree is formed by adding a new node $u$ which become the new parent of $n$ and the (left, right or only according to what $n$ was) child of $p$. The tree is now considered green pointed on $u$. 
\item $G3_n: \cM^\bullet_{n-2} \rightarrow \cM^\star_{n}$
We start from a tree pointed on a node $n$, of parent $p$. We create a new binary node $b$ and a new leaf $\ell$ such that $b$ become the parent of $n$, $p$ become the parent of $b$ and $\ell$ is the right child of $b$. The new tree is considered is blue pointed on $\ell$.
\item $G4_n: \cM^\bullet_{n-2} \rightarrow \cM^\star_{n}$
We start from a tree pointed on a node $n$, of parent $p$. We create a new binary node $b$ and a new leaf $\ell$ such that $b$ become the parent of $n$, $p$ become the parent of $b$ and $\ell$ is the left child of $b$. The new tree is considered is red pointed on $\ell$.
\item $G5_n: \cM^\bullet_{n-2} \rightarrow \cM^\star_{n}$
We start from a tree pointed on a node $n$, of parent $p$. We create a new binary node $b$ and a new leaf $\ell$ such that $b$ become the parent of $n$, $p$ become the parent of $b$ and $\ell$ is the left child of $b$. The new tree is considered is blue pointed on $\ell$.
\end{itemize}

\begin{proof}
The reverse bijection $G^{-1}$ is described as different operators, depending of the type of point of $T \in \cM^\star_{n}$.
\begin{itemize}
\item If $T$ is red pointed on a only child $\ell$, its image is in $\cM^\star_{n-1}$. Let $u$ be the parent of $\ell$, $p$ the parent of $u$.
 $G^{-1}(T)$ is the tree where $u$ is erased, and $\ell$ become the (left, right or only according to what $u$ was) child of $p$. We keep the red point on $\ell$. 
\item If $T$ is blue pointed on a only child $\ell$, its image is in $\cM^\star_{n-1}$. Let $u$ be the parent of $\ell$, $p$ the parent of $u$.
 $G^{-1}(T)$ is the tree where $u$ is erased, and $\ell$ become the (left, right or only according to what $u$ was) child of $p$. We keep the blue point on $\ell$. 
\item If $T$ is red pointed on a right leaf $\ell$, its image is in $\cM^\star_{n-1}$. Let $u$ be the parent of $\ell$, $p$ the parent of $u$ and $n$ the left child of $u$. $G^{-1}(T)$ is the tree where $\ell$ is erased, and $u$ become a unary node of parent $p$ and only child $n$. The tree is considered as green pointed on $u$. 
\item If $T$ is green pointed on a unary node $u$, its image is in $\cM^\bullet_{n-1}$. Let $p$ be the parent of $u$, $n$ its child. 
 $G^{-1}(T)$ is the tree where $u$ is erased, and $n$ become the (left, right or only according to what $u$ was) child of $p$. The tree is considered pointed of $n$
 \item If $T$ is blue pointed on a right leaf $\ell$, its image is in $\cM^\bullet_{n-2}$. Let $b$ be the parent of $\ell$, $p$ the parent of $b$ and $n$ the left child of $b$. 
 $G^{-1}(T)$ is the tree where $\ell$ and $b$ are erased, and $n$ become the (left, right or only according to what $u$ was) child of $p$. The tree is now considered pointed on $n$. 
 \item If $T$ is red pointed on a left leaf $\ell$, its image is in $\cM^\bullet_{n-2}$. Let $b$ be the parent of $\ell$, $p$ the parent of $b$ and $n$ the right child of $b$. 
 $G^{-1}(T)$ is the tree where $\ell$ and $b$ are erased, and $n$ become the (left, right or only according to what $u$ was) child of $p$. The tree is now considered pointed on $n$. 
 \item If $T$ is blue pointed on a left leaf $\ell$, its image is in $\cM^\bullet_{n-2}$. Let $b$ be the parent of $\ell$, $p$ the parent of $b$ and $n$ the right child of $b$. 
 $G^{-1}(T)$ is the tree where $\ell$ and $b$ are erased, and $n$ become the (left, right or only according to what $u$ was) child of $p$. The tree is now considered pointed on $n$. 
\end{itemize}

\end{proof}

\subsection{The repointing operation} \label{sec:repoint}

We now give bijections realizing the two identities between combinatorial
classes:
\begin{align*}
\mathcal B^\star &= \mathcal B^\bullet + \mathcal B\text;\\
\mathcal M^\star &= \mathcal M^\bullet + \mathcal M\text.
\end{align*}
To do this, we start from a color pointed tree~$T$ of~$\mathcal B^\star$ or
$\mathcal M^\star$. Our bijections work by leaving the tree~$T$ intact and
changing the point; for this reason, we call it \emph{repointing} the
tree~$T$.

Let $v$ be a node of~$T$. The node~$v$ can be a \emph{left child}, a
\emph{right child}, or an \emph{only child} (in case its parent is unary). By
convention, the root node of~$T$ is considered a right child. We also order
the ancestors of the pointed node~$v$ from~$v$ to the root, including~$v$
itself.

The repointing operation is defined as follows.
\begin{itemize}
\item If $T$ has a blue pointed leaf, say~$\ell$, let $v$ be the first
ancestor of~$\ell$ that is a left child, or $v = \bot$ if no such ancestor
exists.
\item If $T$ has a red pointed leaf~$\ell$, let $v$ be the first ancestor
of~$\ell$ that is a right child (since the root node is a right child, $v$
always exists).
\item If $T$ has a green pointed unary node~$u$, let $v$ be the unique child
of~$u$.
\end{itemize}
In all cases, the result is the pointed tree $(T, v)$, or the unpointed
tree~$T$ if $v = \bot$. It is thus an element of~$\mathcal B^\bullet +
\mathcal B$ or $\mathcal M^\bullet + \mathcal M$.

\input{figurepointcatalan}
\input{figurepointmotzkin}

\begin{proposition} \label{prop:repoint}
The repointing operation described above is a bijection.
\end{proposition}

\begin{proof}
Let $v$ be either a node of~$T$ or~$\bot$; we prove that there exists a unique
pointing of~$T$ the repointing of which yields the tree $(T,v)$. We
distinguish four cases.
\begin{itemize}
\item If $v = \bot$, the original point must be a blue leaf~$\ell$ with no
left child in its ancestors. This means that $\ell$ is the rightmost leaf
of~$T$.
\item If $v$ is a left child, the original point must be a blue leaf~$\ell$
descendant of~$v$ such that the branch from~$v$ to~$\ell$ contains no left
child. This means that $\ell$ is the rightmost descendant leaf of~$v$.
\item Similarly, if $v$ is a right child, the original point must be the red
leaf~$\ell$, where $\ell$ is the leftmost descendant leaf of~$v$.
\item If $v$ is an only child, the original point must be the green unary
node~$u$, where $u$ is the parent node of~$v$.
\end{itemize}
This completes the proof.
\end{proof}

The advantage of this bijection is that it is easily realized as an algorithm,
as shown in Section~\ref{sec:random}.

\section{Random sampling} \label{sec:random}

All algorithms in this section assume that the following operations on a
node~$v$ of a tree~$T$ can be done in time~$\mathcal O(1)$:
\begin{itemize}
\item determine if $v$ is a left child, right child, or only child, and find
its father;
\item determine if $v$ is a binary node, unary node, or leaf, and find its
children.
\end{itemize}
In addition, it is assumed that drawing a random integer in~$\{1,\dotsc,m\}$
can be done in time $\mathcal O(1)$ for $m = 4$ or $9$.

Practical implementation is discussed in Section~\ref{sec:structure}.

% Une astuce existe où chaque noeud n'a à connaître que son père et son type,
% mais ça rend la chose moins intelligible (pour être précis, les points verts
% doivent être sur des fils uniques plutôt que sur des noeuds unaires), je
% propose de ne pas en parler.

\subsection{Binary trees}

We first present a sampler which is linear in time, space and random bits. The proof will be reuse in the case of Motzkin trees. Then, we slightly modify the algorithm in order to obtain a random bits complexity in $2n+ \Theta (\log^2 n)$, the entropy being $2n-\Theta (\log n )$.

% Le théorème et sa preuve supposent l'existence de deux algos. L'un, algo:b1,
% est l'algo qui tente d'engendrer un arbre de taille 2n+1 par greffes et
% repointages répétés et renvoit FAIL en cas de rejet. L'autre, algo:b,
% appelle algo:b1 tant qu'il renvoit FAIL et renvoit l'arbre engendré.

\begin{algorithm}[ht]%[H]
\caption{A try for random binary tree}
\label{algo:b1}
\KwIn{A size $n$.}
\KwOut{A binary tree of size $2n+1$ or FAIL}
$\text{T}:=$ a blue or red pointed leaf (with probability $1/2$ each)\label{line:start}\\
\textbf{Repeat} $n$ \textbf{times:} \\
\hspace{1cm} $(\text{T}, v) =$ repoint T\\
\hspace{1cm} if $v=\bot$ \textbf{return} FAIL\\
\hspace{1cm}\textbf{Choose one of the following cases:} \\
\hspace{1cm}\textbf{With Probability} $\frac{1}{4}$,\\
\hspace{2cm} $\text{T} := F1(T,v) $\\
\hspace{1cm}\textbf{With Probability} $\frac{1}{4}$,\\
\hspace{2cm} $\text{T} := F2(T,v) $\\
\hspace{1cm}\textbf{With Probability} $\frac{1}{4}$,\\
\hspace{2cm} $\text{T} := F3(T,v) $\\
\hspace{1cm}\textbf{With Probability} $\frac{1}{4}$,\\
\hspace{2cm} $\text{T} := F4(T,v) $\\
\textbf{return} T
\end{algorithm}

\begin{algorithm}[ht]%[H]
\caption{Random binary tree}
\label{algo:b}
\KwIn{A size $n$.}
\KwOut{A binary tree of size $2n+1$}
$\text{T}= $ A try for random binary tree $(n)$\\
\textbf{While} $\text{T} = $ FAIL \textbf{do:} \\
\hspace{1cm} $\text{T}= $ A try for random binary tree $(n)$\\
\textbf{return} T
\end{algorithm}

\begin{remark}
In line~\ref{line:start}  of Algorithm~\ref{algo:b1} we can substitute ``blue or red'' by ``red''. Indeed, the algorithm fails until the starting leaf is red.  
\end{remark}

\begin{theorem} \label{thm:algo}
Algorithm~\ref{algo:b} outputs a uniformly distributed Catalan tree with~$2n+1$
nodes. It runs in average linear time.
\end{theorem}

To prove the theorem, we first establish the following lemmas.

\begin{lemma} \label{lem:proba}
Let $T$ be a color pointed binary tree with $2n+1$ nodes. A run of
Algorithm~\ref{algo:b1} reaches $T$ with probability:
\[\mathbb P(T) = \frac1{2\cdot4^n}.\]
\end{lemma}

% La preuve est un peu laconique. C'est essentiellement une conséquence
% triviale du fait que la greffe et le repointage sont des bijections. Je
% préfère attendre d'avoir le texte de la section sur la greffe pour le faire
% (vous pouvez aussi l'écrire si vous voulez).

\begin{proof}
By induction on $n$.
\end{proof}

\begin{lemma} \label{lem:repoint}
Assume that Algorithm~\ref{algo:b} successfully generated a tree with~$2n+1$
nodes. The generation took $\mathcal O(n)$ operations.
\end{lemma}

\begin{proof}
Let $T$ be the tree generated. Generating~$T$ takes $n$ graft operations and
$n$ repointing operations. One graft operation has complexity $\mathcal O(1)$;
however, repointing costs $\mathcal O(d)$, where $d$ is the distance traveled
by the point during repointing. We prove that the point travels a distance
$\mathcal O(n)$ overall.

Consider the depth in the tree~$T$ being generated of the pointed node. Each
graft operation adds one to the depth, while repointing operations reduce the
depth. As the depth is nonnegative, this means that the point traveled at most
a distance~$n$ in repointing. This concludes the proof.
\end{proof}

\begin{proof}[Proof of Theorem~\ref{thm:algo}]
By Lemma~\ref{lem:proba}, all color pointed trees with~$2n+1$ nodes are
generated with the same probability. Since every tree has $2n+2$ possible
colored points, every tree has the same probability of being generated as
well. This proves uniformity.

To prove the complexity, we first note that the number of binary trees with
$2n+1$ nodes is the Catalan number~$C_n$, which is asymptotically on the order
of $4^nn^{-3/2}$. By Lemma~\ref{lem:proba}, since every tree has $2n+2$
possible colored pointings, the probability for a run of
Algorithm~\ref{algo:b1} to reach size $2n+1$ is:
\[\frac{(2n+2)C_n}{2\cdot4^n}\approx n^{-1/2}.\]
This means that there are, on average, $\mathcal O(\sqrt n)$ failed runs.

Let us now find the expected cost of a failed run, by computing the
probability of failing at size $2k+1$ with $k < n$. By
Proposition~\ref{prop:repoint}, only one colored pointing of every tree with
$2k+1$ nodes causes Algorithm~\ref{algo:b1} to fail at the next step. By
Lemma~\ref{lem:proba}, the probability of failing at size~$2k+1$ is thus:
\[\frac{C_k}{2\cdot4^k}\approx k^{-3/2}.\]
By Lemma~\ref{lem:repoint}, the expected cost of a failed run is therefore:
\[\sum_{k=0}^{n-1} \mathcal O(k\cdot k^{-3/2}) = \mathcal O(\sqrt n).\]
Putting this together, generating a tree with $2n+1$ nodes involves, on
average, $\mathcal O(\sqrt n)$ runs of Algorithm~\ref{algo:b1}, each costing
$\mathcal O(\sqrt n)$ on average. Therefore, the average total cost is
$\mathcal O(n)$.
\end{proof}

The next algorithm is an improved version of Algorithm~\ref{algo:b}. It is
illustrated in Figure~\ref{fig:exrun}.

\begin{algorithm}[ht]%[H]
\caption{Efficient sampler of binary tree}
\label{algo:efficient}
\KwIn{A size $n$.}
\KwOut{A binary tree of size $2n+1$}
$\text{T}:=$ a red pointed leaf (with probability $1/2$ each)\\
\textbf{Repeat} $n$ \textbf{times:} \\
\hspace{1cm} $(\text{T}, v) =$ repoint T\\
\hspace{1cm} if $v=\bot$\\
\hspace{2cm} $v:=$ a uniformly chosen node of T. \\
\hspace{1cm}\textbf{Choose one of the following cases:} \\
\hspace{1cm}\textbf{With Probability} $\frac{1}{4}$,\\
\hspace{2cm} $\text{T} := F1(T,v) $\\
\hspace{1cm}\textbf{With Probability} $\frac{1}{4}$,\\
\hspace{2cm} $\text{T} := F2(T,v) $\\
\hspace{1cm}\textbf{With Probability} $\frac{1}{4}$,\\
\hspace{2cm} $\text{T} := F3(T,v) $\\
\hspace{1cm}\textbf{With Probability} $\frac{1}{4}$,\\
\hspace{2cm} $\text{T} := F4(T,v) $\\
\textbf{return} T
\end{algorithm}

\begin{theorem} \label{thm:efficient}
Algorithm~\ref{algo:efficient} draws uniformly a tree of size~$2n+1$ in
average time and space complexity $O(n)$. Moreover, let $B_n$ be the random
variable representing the number of random bits required by the algorithm; we
have on average:
\[B_n = 2n + \frac{\log^2n}{4\log2} + \mathrm o(\log^2n).\]
\end{theorem}

This result means that the number of random bits required is close to the
entropy~$2n$ in a very strong sense: the excess bits are of the order of
$\log^2n$, making the algorithm nearly optimal in this respect.

\begin{proof}
To correction of the algorithm is easily shown by remarking that, at each
step~$i$, the tree~$T$ is uniformly distributed among the colored pointed
trees of size~$2i+1$.

To show the complexity, we first note that every iteration of the loop
consumes $2$ random bits to choose between the four graft operations. This
adds up to $2n$ random bits. Now, consider a step~$i$. Since the tree~$T$ is
uniformly distributed and since, for every unpointed tree, only one colored
pointing causes $v$ to be equal to $\bot$, the probability of having to
randomly repoint the tree is $1/(2i+2)$. Moreover, since there are $2i+1$
possible choices for~$v$, the cost in random bits
is~$\lfloor\log_2(2i+1)\rfloor$. By summing over all steps~$i$, the average
number of random bits consumed in this way is:
\[\sum_{i=0}^{n-1} \frac{\lfloor\log_2(2i+1)\rfloor}{2i+2}
\sim\int_0^n\frac{\log(2t+1)}{(2t+2)\log 2}dt \sim \frac{\log^2 n}{4\log 2}.\]
We thus get the announced estimate.

Finally, we prove that the average time complexity is linear. The only
operation that might not be is the repointing. As in the proof of
Lemma~\ref{lem:repoint}, we consider the depth of the pointed node in the
tree~$T$ being genereated. This time, the depth is not only increased by the
graft operations, but also by the random repointings.

At step~$i$, there is a probability $1/(2i+2)$ of randomly repointing; after
this operation, since the tree~$T$ is random, the pointed node will have depth
on the order of~$\sqrt i$ on average. Thus, the total depth increase due to
the random repointings is of the order of:
\[\sum_{i=0}^{n-1}\frac{\sqrt i}{2i+2} \sim \mathcal O(\sqrt n).\]
As the depth of the pointed node is nonnegative, the point cannot travel more
than $n + \mathcal O(\sqrt n)$ nodes during repointing, which shows that the
generation takes linear time on average.
\end{proof}

\input{exemplecatalanmieux}

\subsection{Motzkin trees}

% Je suis passé très vite sur les points de preuve qui sont identiques au cas
% binaire. Trop vite ?

\begin{algorithm}%[H]
\caption{A try for random unary-binary tree}
\label{algo:m1}
\KwIn{A size $n$.}
\KwOut{A unary-binary tree of size $n$, $n+1$ or FAIL}
$\text{T}:=$ a blue or red pointed leaf (with probability $1/2$ each)\\
\textbf{While} $\abs{\text{T}}<n$ \textbf{do:} \\
\hspace{.5cm}\textbf{Choose one of the following cases:} \\
\hspace{1cm}\textbf{With Probability} $\frac{1}{3}$,\\
\hspace{1.5cm} $\text{T} := G1(T) $\\
\hspace{1cm}\textbf{With Probability} $\frac{1}{3}$,\\
\hspace{1.5cm} $(\text{T}, v) =$ repoint T\\
\hspace{1.5cm} if $v=\bot$ \textbf{return} FAIL\\
\hspace{1.5cm} $\text{T} := G2(T,v) $\\
\hspace{1cm}\textbf{With Probability} $\frac{1}{9}$,\\
\hspace{1.5cm} $(\text{T}, v) =$ repoint T\\
\hspace{1.5cm} if $v=\bot$ \textbf{return} FAIL\\
\hspace{1.5cm} $\text{T} := G3(T,v) $\\
\hspace{1cm}\textbf{With Probability} $\frac{1}{9}$,\\
\hspace{1.5cm} $(\text{T}, v) =$ repoint T\\
\hspace{1.5cm} if $v=\bot$ \textbf{return} FAIL\\
\hspace{1.5cm} $\text{T} := G4(T,v) $\\
\hspace{1cm}\textbf{With Probability} $\frac{1}{9}$,\\
\hspace{1.5cm} $(\text{T}, v) =$ repoint T\\
\hspace{1.5cm} if $v=\bot$ \textbf{return} FAIL\\
\hspace{1.5cm} $\text{T} := G5(T,v) $\\
\textbf{return} T
\end{algorithm}

\begin{algorithm}%[H]
\caption{Random unary-binary tree}
\label{algo:m}
\KwIn{A size $n$.}
\KwOut{A unary-binary tree of size $n$}
$\text{T}= $ A try for random binary tree $(n)$\\
\textbf{While} $\text{T} = $ FAIL or $\abs{T}=n+1$ \textbf{do:} \\
\hspace{1cm} $\text{T}= $ A try for random binary tree $(n)$\\
\textbf{return} T
\end{algorithm}

\begin{theorem} \label{thm:algom}
Algorithm~\ref{algo:m} outputs a uniformly distributed Motzkin tree with~$n$ nodes. It
runs in average linear time.
\end{theorem}

The proof is very similar to that of Theorem~\ref{thm:algo}. In the following,
we skim over the similar parts to focus on the differences. We start by
stating two lemmas that are analogues to Lemmas~\ref{lem:proba} and
\ref{lem:repoint}.

\begin{lemma} \label{lem:probam}
Let $T$ be a color pointed unary-binary tree with $n$ nodes. A run of
Algorithm~\ref{algo:m1} reaches $T$ with probability:
\[\mathbb P(T) = \frac1{2\cdot3^{n-1}}.\]
\end{lemma}

\begin{proof}
By induction on $n$.
\end{proof}

\begin{lemma} \label{lem:repointm}
Assume that Algorithm~\ref{algo:m} successfully generated a tree with~$n$
nodes. The generation took $\mathcal O(n)$ operations.
\end{lemma}

\begin{proof}
Like in the proof of Lemma~\ref{lem:repoint}, we need to prove that, in the
course of the algorithm, the point only travels a distance $\mathcal O(n)$
during repointing. We again consider the depth of the pointed node in the
tree~$T$ being generated. This depth:
\begin{itemize}
\item increases by one every time one of the operations $G1$, $G3$, $G4$ or
$G5$ is applied;
\item increases by one every time a green pointed node is repointed;
\item stays the same if $G2$ is applied;
\item decreases if the point travels upward during repointing.
\end{itemize}
Moreover, a green node can only be produced by $G2$, which means that the
depth of the point increases at most one unit per iteration. Since the depth
is nonnegative, this means that the point can only travel upward $n$~times at
most.
\end{proof}

\begin{proof}[Proof of Theorem~\ref{thm:algo}]
By Lemma~\ref{lem:probam}, all color pointed trees with~$n$ nodes are
generated with the same probability; since every tree has $n+1$ possible
pointings, uniformity is proven.

The number of Motzkin trees with $n$ nodes is the $n$th Motzkin number $M_n$,
which is on the order of~$3^nn^{-3/2}$. Using Lemma~\ref{lem:probam}, we
prove, in a manner identical to the proof of Theorem~\ref{thm:algo}, that the
probability of success of a run of Algorithm~\ref{algo:m1} is:
\[\frac{(n+1)M_n}{2\cdot3^{n-1}}\approx n^{-1/2},\]
so that there are on average $\mathcal O(\sqrt n)$ failed runs.

Let us find the average cost of a failed run. If $k < n$,
Proposition~\ref{prop:repoint} and \ref{lem:probam} show that the probability
of failing at size $k$ is:
\[\frac{M_k}{2\cdot 3^{k-1}}\approx k^{-3/2}.\]
Moreover, there are $(n+2)M_{n+1}$ color pointed trees of size~$n+1$. By
Lemma~\ref{lem:probam}, the probability of failing at size~$n+1$ is at most:
\[\frac{(n+2)M_{n+1}}{2\cdot3^{n-1}}\approx n^{-1/2}.\]
By Lemma~\ref{lem:repointm}, the average cost of a failed run is thus:
\[\sum_{k=0}^{n-1} \mathcal O(k\cdot k^{-3/2}) + \mathcal O(n\cdot n^{-1/2}) =
\mathcal O(\sqrt n).\]
The total average cost of Algorithm~\ref{algo:m} is thus $\mathcal O(\sqrt
n)$ times $\mathcal O(\sqrt n)$, or $\mathcal O(n)$.
\end{proof}

\subsection{Data structures and practical implementation} \label{sec:structure}

To implement our algorithm, we represented trees as an array of nodes, each
node consisting of four values:
\begin{itemize}
\item one flag representing its arity (leaf, unary, or binary),
\item the index of its parent (with a special value for the root),
\item two indices for its eventual children.
\end{itemize}
Thus, a node requires 128~bits of storage if the tree has less than $2^{32}$
nodes. With this information, all required operations can be done in constant
time.

Our implementation is able to sample a tree of size 10~million in about
5~seconds on a standard PC.

% graph : nombre experimental de nombre de bits utilises pour generer des arbres de taille 50 à un million

\section{Weighted specification} \label{sec:weight}
In order to simulate other distribution we might want to put different weights to different types of node. We present in this section how to adapt the bijection and the random sampler for the class $\cA$ of unary-binary trees with a weight $\cU$ on unary nodes. 
\begin{align}
\label{eq:biv}
\cA^\star=2\cZ + \cZ \cU \cA^\star + \cZ \cU \cA^\bullet + 4\cZ ^2 \cA^\bullet - \cZ ^2 \cU ^2 \cA^\bullet
\end{align}

In fact, the negative term $\cZ ^2 \cU ^2 \cM^\bullet$ can be rewritten as $\cZ \cU (\cZ \cU \cM^\bullet) $ which, in the bijection for Motzkin tree correspond to the operation of creating a green pointed unary node and then transform it in a binary node. So the equation~\ref{eq:biv} become:

\begin{align*}
\cA^\star=2\cZ + \cZ \cU (\cA^{\textcolor{blue}{(\ell)}}+\cA^{\textcolor{red}{(\ell)}}) + \cZ \cU \cA^\bullet + 4\cZ ^2 \cA^\bullet
\end{align*}

where  $\cA^{\textcolor{blue}{(\ell)}}$ is the class of weighted unary-binary tree blue pointed on a leaf and $\cA^{\textcolor{red}{(\ell)}}$ the class of weighted unary-binary tree red pointed on a leaf.

From this, we describe the bijection $H: 2\cZ + \cZ \cU (\cA^{\textcolor{blue}{(\ell)}}+\cA^{\textcolor{red}{(\ell)}})  + \cZ \cU \cA^\bullet + 4\cZ ^2 \cA^\bullet \rightarrow \cA^\star $, very similar to $G$:

\begin{itemize}
\item $H1_{n,k}: \cA^{\textcolor{red}{(\ell)}}_{n-1,k-1} \rightarrow \cA^\star_{n,k}$
We start from a tree red pointed on a leaf $\ell$, we call $p$ the parent of $\ell$. The operation $G1_n$ consists in adding a new unary node as the new parent of $\ell$ and (left, right or only according to what $\ell$ was) child of $p$. We keep the red point on $\ell$. 
\item $H2_{n,k}: \cA^{\textcolor{blue}{(\ell)}}_{n-1,k-1} \rightarrow \cA^\star_{n,k}$
We start from a tree blue pointed on a leaf $\ell$, we do a similar modification and keep the blue point on $\ell$. 
\item $H3_{n,k}: \cA^\bullet_{n-1,k-1} \rightarrow \cA^\star_{n}$
We start from a tree pointed on a node $n$, of parent $p$. The new tree is formed by adding a new node $u$ which become the new parent of $n$ and the (left, right or only according to what $n$ was) child of $p$. The tree is now considered green pointed on $u$. 
\item $H4_{n,k}: \cA^\bullet_{n-2, k} \rightarrow \cA^\star_{n,k}$
We start from a tree pointed on a node $n$, of parent $p$. We create a new binary node $b$ and a new leaf $\ell$ such that $b$ become the parent of $n$, $p$ become the parent of $b$ and $\ell$ is the right child of $b$. The new tree is considered is blue pointed on $\ell$.
\item $H5_{n,k}: \cA^\bullet_{n-2,k} \rightarrow \cA^\star_{n,k}$
We start from a tree pointed on a node $n$, of parent $p$. We create a new binary node $b$ and a new leaf $\ell$ such that $b$ become the parent of $n$, $p$ become the parent of $b$ and $\ell$ is the left child of $b$. The new tree is considered is red pointed on $\ell$.
\item $H6_{n,k}: \cA^\bullet_{n-2,k} \rightarrow \cA^\star_{n,k}$
We start from a tree pointed on a node $n$, of parent $p$. We create a new binary node $b$ and a new leaf $\ell$ such that $b$ become the parent of $n$, $p$ become the parent of $b$ and $\ell$ is the left child of $b$. The new tree is considered is blue pointed on $\ell$.
\item $H7_{n,k}: \cA^\bullet_{n-2, k} \rightarrow \cA^\star_{n,k}$
We start from a tree pointed on a node $n$, of parent $p$. We create a new binary node $b$ and a new leaf $\ell$ such that $b$ become the parent of $n$, $p$ become the parent of $b$ and $\ell$ is the right child of $b$. The new tree is considered is red pointed on $\ell$.
\end{itemize}

With this specifiaction we can immediately obtain a random sampler, similar to Algorithm~\ref{algo:m1} with a extra rejection when we want to grow a green pointed tree. However, this algorithm will be in average complexity exponential and so is not as efficient than classical algorithms (like for example, bivariate Boltzmann samplers \cite{bodiniponty}, anticipated rejection \cite{Barcucci19953} or methods via Łukasiewicz words \cite{bodinicissedavid}). A perspective work of the authors is to tune this sampler in order to obtain a good complexity for weighted specifications. 
% G6 : le cas qui venait de G1
% On a donc 1 + u fois ceux comme ca
% plus de rejet

\section{Conclusion}
%Bijections based on holonomic equations enable us to have a local growing process on trees instead of a branching one, which leads to an efficient sampling. As every algebric generating function is also holonomic, we hope to be able to find more bijections of this kind. In the long term, we can imagine defining a whole family of combinatorial objects which can be drawn using Rémy-like samplers. 
%%Bit-otimality
%
%Moreover, in case of Catalan tree, we can sample with a optimum number of random bits plus $O(\log^2 n)$. We are currently investigating other classes for which we might describe a nearly optimum sampler.    

We presented two new uniform random samplers for binary and
unary-binary trees which respectively consume $2n+\Theta(\ln(n)^2)$
(resp. $\Theta(n)$) random bits in average. These samplers are based on
holonomic specifications that replace the branching process with a local growing process (grafting process).

As every algebraic generating function is also holonomic, we can
imagine defining a whole
family of combinatorial objects which can be drawn using grafting
processes. The main drawback is the need to be able to find a combinatorial
interpretation for the holonomic equations. It is not necessary always
possible, and even for simple combinatorial objects this is not
elementary. The first extensions we expect to tackle are the classes
of constrained words, such as the class of words on the alphabet
$\{a,b\}$ having the same number of $a$ and $b$ and all the variations
around this class. Other classes bring traditionally holonomic
equations, such as constrained paths, meanders or permutations with
forbidden patterns. Each of these classes presents a challenge, and finding an optimal generator for any of them would be an important step for the general framework of random sampling. 

\bibliographystyle{plain}
\bibliography{biblio}

\end{document}

%% file: bijremyaxel.tex
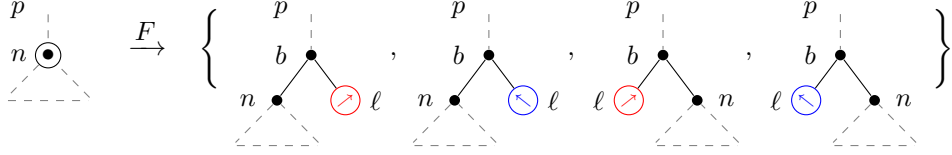
\begin{figure}
\begin{center}
\begin{tabular}{ccc}
\begin{tikzpicture}[baseline=(n),inner sep=1pt, circle, minimum size=0pt, xscale=0.75, scale=.6]
\node (p) at (0,0){};
\node at (p)[left=2mm] {$p$};
\node [draw, minimum size=5pt] (n) at (0,-1){$\bullet$};
\node at (n)[left=2mm] {$n$};
\node (S1) at (-1.3,-2.0) {};
\node (S2) at (1.3,-2.0) {};
\node (chut) at (0,-2.1) {};
\draw [dashed, gray] (p) to (n);
\draw [dashed, gray, minimum size=2pt] (n)--(S1)--(S2)--(n)--cycle;
\end{tikzpicture}&
$\xrightarrow{\displaystyle F}$&
$\Biggl\lbrace$ 
\begin{tikzpicture}[baseline=(b),inner sep=1pt, circle, minimum size=0pt, xscale=0.75, scale=.6]
\node (p) at (1,1){};
\node at (p)[left=2mm] {$p$};
\node [fill, minimum size=4pt] (b) at (1,0){};
\node at (b)[left=2mm] {$b$};
\node [red, draw, minimum size=4pt] (l) at (2,-1){\tiny $\nearrow$};
\node at (l)[right=2mm] {$\ell$};
\node [fill, minimum size=4pt] (n) at (0,-1){};
\node at (n)[left=2mm] {$n$};
\node (S1) at (-1.3,-2) {};
\node (S2) at (1.3,-2) {};
\node (chut) at (0,-2.1) {};
\draw (b) to (n);
\draw (b) to (l);
\draw [dashed, gray] (p) to (b);
\draw [dashed, gray, minimum size=2pt] (n)--(S1)--(S2)--(n)--cycle;
\end{tikzpicture}, 
\begin{tikzpicture}[baseline=(b),inner sep=1pt, circle, minimum size=0pt, xscale=0.75, scale=.6]
\node (p) at (1,1){};
\node at (p)[left=2mm] {$p$};
\node [fill, minimum size=4pt] (b) at (1,0){};
\node at (b)[left=2mm] {$b$};
\node [blue, draw, minimum size=4pt] (l) at (2,-1){\tiny $\nwarrow$};
\node at (l)[right=2mm] {$\ell$};
\node [fill, minimum size=4pt] (n) at (0,-1){};
\node at (n)[left=2mm] {$n$};
\node (S1) at (-1.3,-2) {};
\node (S2) at (1.3,-2) {};
\node (chut) at (0,-2.1) {};
\draw (b) to (n);
\draw (b) to (l);
\draw [dashed, gray] (p) to (b);
\draw [dashed, gray, minimum size=2pt] (n)--(S1)--(S2)--(n)--cycle;
\end{tikzpicture},
\begin{tikzpicture}[baseline=(b),inner sep=1pt, circle, minimum size=0pt, xscale=0.75, scale=.6]
\node (p) at (1,1){};
\node at (p)[left=2mm] {$p$};
\node [fill, minimum size=4pt] (b) at (1,0){};
\node at (b)[left=2mm] {$b$};
\node [red, draw, minimum size=4pt] (l) at (0,-1){\tiny $\nearrow$};
\node at (l)[left=2mm] {$\ell$};
\node [fill, minimum size=4pt] (n) at (2,-1){};
\node at (n)[right=2mm] {$n$};
\node (S1) at (.7,-2) {};
\node (S2) at (3.3,-2) {};
\node (chut) at (0,-2.1) {};
\draw (b) to (n);
\draw (b) to (l);
\draw [dashed, gray] (p) to (b);
\draw [dashed, gray, minimum size=2pt] (n)--(S1)--(S2)--(n)--cycle;
\end{tikzpicture},
\begin{tikzpicture}[baseline=(b),inner sep=1pt, circle, minimum size=0pt, xscale=0.75, scale=.6]
\node (p) at (1,1){};
\node at (p)[left=2mm] {$p$};
\node [fill, minimum size=4pt] (b) at (1,0){};
\node at (b)[left=2mm] {$b$};
\node [blue, draw, minimum size=4pt] (l) at (0,-1){\tiny $\nwarrow$};
\node at (l)[left=2mm] {$\ell$};
\node [fill, minimum size=4pt] (n) at (2,-1){};
\node at (n)[right=2mm] {$n$};
\node (S1) at (.7,-2) {};
\node (S2) at (3.3,-2) {};
\node (chut) at (0,-2.1) {};
\draw (b) to (n);
\draw (b) to (l);
\draw [dashed, gray] (p) to (b);
\draw [dashed, gray, minimum size=2pt] (n)--(S1)--(S2)--(n)--cycle;
\end{tikzpicture}
$\Biggr\rbrace$
\end{tabular}
\end{center}
\caption{The bijection from $\cB^{\bullet}_n$ to $\cB^{\star}_{n+2}$. Note
that the node~$b$ takes the place of~$n$: it becomes the root if $n$ was the
root, and the left (\textit{resp.}\ right) child of~$p$ if $n$ was the left
(\textit{resp.}\ right) child of~$p$.}
\label{fig:bijremy}
\end{figure}

%% file: bijmotzaxel.tex
\begin{figure}
\begin{center}
\begin{tabular}{ccc|ccc|ccc}
\begin{tikzpicture}[baseline=(l),inner sep=1pt, circle, minimum size=0pt, xscale=0.75, scale=.6]
\node at (0,.5) {};
\node (p) at (0,0){};
\node at (p) [left=2mm] {$p$};
\node [blue, draw, minimum size=4pt] (l) at (0,-1){\tiny $\nwarrow$};
\node at (l) [left=2mm] {$\ell$};
\node (chut) at (0,-1.1) {};
\draw [dashed, gray] (p) to (l);
\end{tikzpicture}&
$\xrightarrow{\displaystyle G1}$&
\begin{tikzpicture}[baseline=(u),inner sep=1pt, circle, minimum size=0pt, xscale=0.75, scale=.6]
\node (p) at (0,0){};
\node at (p)[left=2mm] {$p$};
\node [fill, minimum size=4pt] (u) at (0,-1) {};
\node at (u)[left=2mm] {$u$};
\node [blue, draw, minimum size=4pt] (l) at (0,-2){\tiny $\nwarrow$};
\node at (l) [left=2mm] {$\ell$};
\node (chut) at (0,-2.1) {};
\draw [dashed, gray] (p) to (u);
\draw  (l) to (u);
\end{tikzpicture}&
\begin{tikzpicture}[baseline=(l),inner sep=1pt, circle, minimum size=0pt, xscale=0.75, scale=.6]
\node (p) at (0,0){};
\node at (p)[left=2mm] {$p$};
\node [red, draw, minimum size=4pt] (l) at (0,-1){\tiny $\nearrow$};
\node at (l) [left=2mm] {$\ell$};
\node (chut) at (0,-1.1) {};
\draw [dashed, gray] (p) to (l);
\end{tikzpicture}&
$\xrightarrow{\displaystyle G1}$&
\begin{tikzpicture}[baseline=(u),inner sep=1pt, circle, minimum size=0pt, xscale=0.75, scale=.6]
\node (p) at (0,0){};
\node at (p)[left=2mm] {$p$};
\node [fill, minimum size=4pt] (u) at (0,-1) {};
\node at (u)[left=2mm] {$u$};
\node [red, draw, minimum size=4pt] (l) at (0,-2){\tiny $\nearrow$};
\node at (l) [left=2mm] {$\ell$};
\node (chut) at (0,-2.1) {};
\draw [dashed, gray] (p) to (u);
\draw  (l) to (u);
\end{tikzpicture}&
\begin{tikzpicture}[baseline=(u),inner sep=1pt, circle, minimum size=0pt, xscale=0.75, scale=.6]
\node (p) at (0,1){};
\node at (p)[left=3mm] {$p$};
\node [green, draw, minimum size=4pt](u) at (0,0){$\downarrow$};
\node at (u)[left=3mm] {$u$};
\node [fill, minimum size=4pt] (n) at (0,-1){};
\node at (n)[left=3mm] {$n$};
\node (S1) at (-1.3,-2) {};
\node (S2) at (1.3,-2) {};
\node (chut) at (0,-2.1) {};
\draw [dashed, gray] (p) to (u);
\draw (n) to (u);
\draw [dashed, gray, minimum size=2pt] (n)--(S1)--(S2)--(n)--cycle;
\end{tikzpicture}&
$\xrightarrow{\displaystyle G1}$&
\begin{tikzpicture}[baseline=(b),inner sep=1pt, circle, minimum size=0pt, xscale=0.75, scale=.6]
\node (p) at (1,1){};
\node at (p)[left=2mm] {$p$};
\node [fill, minimum size=4pt] (b) at (1,0){};
\node at (b)[left=2mm] {$u$};
\node [red, draw, minimum size=4pt] (l) at (2,-1){\tiny $\nearrow$};
\node at (l) [right=2mm] {$\ell$};
\node [fill, minimum size=4pt] (n) at (0,-1){};
\node at (n)[left=2mm] {$n$};
\node (S1) at (-1.3,-2) {};
\node (S2) at (1.3,-2) {};
\node (chut) at (0,-2.1) {};
\draw (b) to (n);
\draw (b) to (l);
\draw [dashed, gray] (p) to (b);
\draw [dashed, gray, minimum size=2pt] (n)--(S1)--(S2)--(n)--cycle;
\end{tikzpicture}
\end{tabular}\\\hrule
\begin{tabular}{ccc}
\begin{tikzpicture}[baseline=(n),inner sep=1pt, circle, minimum size=0pt, xscale=0.75, scale=.6]
\node at (0,.5) {};
\node (p) at (0,0){};
\node at (p)[left=2mm] {$p$};
\node [draw, minimum size=5pt] (n) at (0,-1){$\bullet$};
\node at (n)[left=2mm] {$n$};
\node (S1) at (-1.3,-2) {};
\node (S2) at (1.3,-2) {};
\node (chut) at (0,-2.1) {};
\draw [dashed, gray] (p) to (n);
\draw [dashed, gray, minimum size=2pt] (n)--(S1)--(S2)--(n)--cycle;
\end{tikzpicture}&
$\xrightarrow{\displaystyle G2}$&
\begin{tikzpicture}[baseline=(u),inner sep=1pt, circle, minimum size=0pt, xscale=0.75, scale=.6]
\node (p) at (0,1){};
\node at (p)[left=3mm] {$p$};
\node [green, draw, minimum size=4pt](u) at (0,0){$\downarrow$};
\node at (u)[left=3mm] {$u$};
\node [fill, minimum size=4pt] (n) at (0,-1){};
\node at (n)[left=3mm] {$n$};
\node (S1) at (-1.3,-2) {};
\node (S2) at (1.3,-2) {};
\node (chut) at (0,-2.1) {};
\draw [dashed, gray] (p) to (u);
\draw (n) to (u);
\draw [dashed, gray, minimum size=2pt] (n)--(S1)--(S2)--(n)--cycle;
\end{tikzpicture}
\end{tabular}\\\hrule
\begin{tabular}{ccc}
\begin{tikzpicture}[baseline=(n),inner sep=1pt, circle, minimum size=0pt, xscale=0.75, scale=.6]
\node at (0,.5) {};
\node (p) at (0,0){};
\node at (p)[left=2mm] {$p$};
\node [draw, minimum size=5pt] (n) at (0,-1){$\bullet$};
\node at (n)[left=2mm] {$n$};
\node (S1) at (-1.3,-2) {};
\node (S2) at (1.3,-2) {};
\node (chut) at (0,-2.1) {};
\draw [dashed, gray] (p) to (n);
\draw [dashed, gray, minimum size=2pt] (n)--(S1)--(S2)--(n)--cycle;
\end{tikzpicture}&
$\xrightarrow{\displaystyle\{G3,G4,G5\}}$&
$\Biggl\lbrace$
\begin{tikzpicture}[baseline=(b),inner sep=1pt, circle, minimum size=0pt, xscale=0.75, scale=.6]
\node (p) at (1,1){};
\node at (p)[left=2mm] {$p$};
\node [fill, minimum size=4pt] (b) at (1,0){};
\node at (b)[left=2mm] {$b$};
\node [blue, draw, minimum size=4pt] (l) at (2,-1){\tiny $\nwarrow$};
\node at (l)[right=2mm] {$\ell$};
\node [fill, minimum size=4pt] (n) at (0,-1){};
\node at (n)[left=2mm] {$n$};
\node (S1) at (-1.3,-2) {};
\node (S2) at (1.3,-2) {};
\node (chut) at (0,-2.1) {};
\draw (b) to (n);
\draw (b) to (l);
\draw [dashed, gray] (p) to (b);
\draw [dashed, gray, minimum size=2pt] (n)--(S1)--(S2)--(n)--cycle;
\end{tikzpicture},
\begin{tikzpicture}[baseline=(b),inner sep=1pt, circle, minimum size=0pt, xscale=0.75, scale=.6]
\node (p) at (1,1){};
\node at (p)[left=2mm] {$p$};
\node [fill, minimum size=4pt] (b) at (1,0){};
\node at (b)[left=2mm] {$b$};
\node [red, draw, minimum size=4pt] (l) at (0,-1){\tiny $\nearrow$};
\node at (l)[left=2mm] {$\ell$};
\node [fill, minimum size=4pt] (n) at (2,-1){};
\node at (n)[right=2mm] {$n$};
\node (S1) at (.7,-2) {};
\node (S2) at (3.3,-2) {};
\node (chut) at (0,-2.1) {};
\draw (b) to (n);
\draw (b) to (l);
\draw [dashed, gray] (p) to (b);
\draw [dashed, gray, minimum size=2pt] (n)--(S1)--(S2)--(n)--cycle;
\end{tikzpicture},
\begin{tikzpicture}[baseline=(b),inner sep=1pt, circle, minimum size=0pt, xscale=0.75, scale=.6]
\node (p) at (1,1){};
\node at (p)[left=2mm] {$p$};
\node [fill, minimum size=4pt] (b) at (1,0){};
\node at (b)[left=2mm] {$b$};
\node [blue, draw, minimum size=4pt] (l) at (0,-1){\tiny $\nwarrow$};
\node at (l)[left=2mm] {$\ell$};
\node [fill, minimum size=4pt] (n) at (2,-1){};
\node at (n)[right=2mm] {$n$};
\node (S1) at (.7,-2) {};
\node (S2) at (3.3,-2) {};
\node (chut) at (0,-2.1) {};
\draw (b) to (n);
\draw (b) to (l);
\draw [dashed, gray] (p) to (b);
\draw [dashed, gray, minimum size=2pt] (n)--(S1)--(S2)--(n)--cycle;
\end{tikzpicture}
\hspace{-1em}
$\Biggr\rbrace$
\end{tabular}
\end{center}
\caption{The bijection $G$ for Motzkin trees. In the case $G1$ (above), the
result depends on the color of the pointed node. Note that applying $G1$ or
$G2$ increases the size by~one, while applying $G3$, $G4$ or $G5$ increases
the size by~two.
% $\cB^{\bullet}_n$ to $\cB^{\star}_{n+2}$. 
}
\label{fig:bijmotz}
\end{figure}
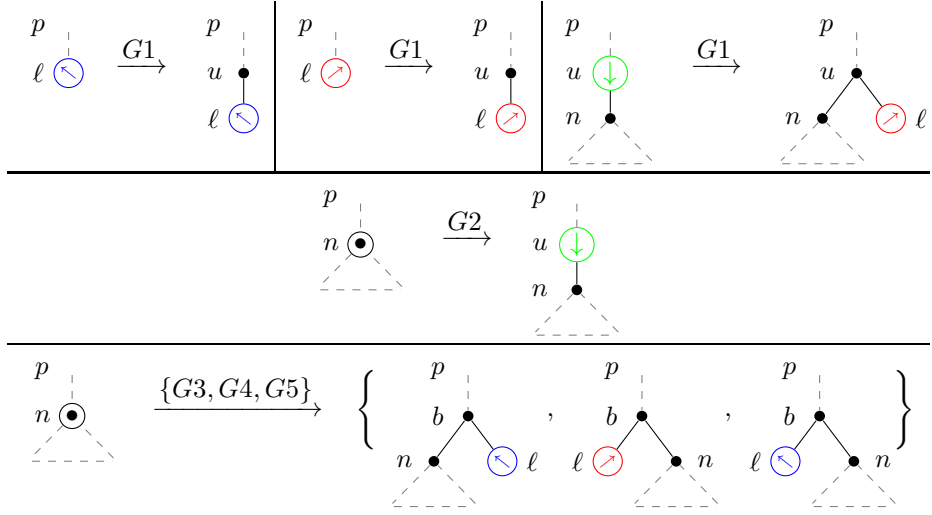

%% file: figurepointcatalan.tex
\begin{figure}
\begin{center}
\begin{tikzpicture}[baseline=(2),inner sep=1pt, circle, minimum size=0pt, xscale=.9, scale =0.7]\node (p) at (0,0){};
\node(0) at  (1.4,4.6) {};
\node (1) at (2,4){$\bullet$};
\node (2) at (1,3){$\bullet$};
\node (3) at (.4,2){$\bullet$};
\node (4) at (1.6,2){$\bullet$};
\node (5) at (1.1,1){$\bullet$};
\node (6) at (.7,0){$\bullet$};
\node (7) at (1.4,0){$\bullet$};
\node (8) at (2.1,1){$\bullet$};
\node (9) at (1.7,0){$\bullet$};
\node[blue, draw, minimum size=4pt] (10) at (2.4,0){\tiny $\nwarrow$};
\node (11) at (3,3){$\bullet$};
\node (12) at (2.4,2){$\bullet$};
\node (13) at (3.6,2){$\bullet$};
\draw [dotted ] (0) to (1);
%\draw (0) to (1);
\draw (2) to (1);
\draw (11) to (1);
\draw (2) to (3);
\draw (2) to (4);
\draw (4) to (5);
\draw (4) to (8);
\draw (5) to (6);
\draw (5) to (7);
\draw (8) to (9);
\draw (8) to (10);
\draw (11) to (12);
\draw (11) to (13);
\end{tikzpicture}
\hspace{-1mm}$\rightsquigarrow$\hspace{-1mm}%
\begin{tikzpicture}[baseline=(2),inner sep=1pt, circle, minimum size=0pt, xscale=.9, scale =0.7]\node (p) at (0,0){};
\node(0) at  (1.4,4.6) {};
\node (1) at (2,4){$\bullet$};
\node[draw, minimum size=4pt] (2) at (1,3){$\bullet$};
\node (3) at (.4,2){$\bullet$};
\node (4) at (1.6,2){$\bullet$};
\node (5) at (1.1,1){$\bullet$};
\node (6) at (.7,0){$\bullet$};
\node (7) at (1.4,0){$\bullet$};
\node (8) at (2.1,1){$\bullet$};
\node (9) at (1.7,0){$\bullet$};
\node (10) at (2.4,0){$\bullet$};
\node (11) at (3,3){$\bullet$};
\node (12) at (2.4,2){$\bullet$};
\node (13) at (3.6,2){$\bullet$};
\draw [dotted ] (0) to (1);
%\draw (0) to (1);
\draw (2) to (1);
\draw (11) to (1);
\draw (2) to (3);
\draw (2) to (4);
\draw (4) to (5);
\draw (4) to (8);
\draw (5) to (6);
\draw (5) to (7);
\draw (8) to (9);
\draw (8) to (10);
\draw (11) to (12);
\draw (11) to (13);
\end{tikzpicture}
\hspace{.1cm}
\begin{tikzpicture}[baseline=(2),inner sep=1pt, circle, minimum size=0pt, xscale=.9, scale =0.7]\node (p) at (0,0){};
\node(0) at  (1.4,4.6) {};
\node (1) at (2,4){$\bullet$};
\node (2) at (1,3){$\bullet$};
\node (3) at (.4,2){$\bullet$};
\node (4) at (1.6,2){$\bullet$};
\node (5) at (1.1,1){$\bullet$};
\node (6) at (.7,0){$\bullet$};
\node (7) at (1.4,0){$\bullet$};
\node (8) at (2.1,1){$\bullet$};
\node (9) at (1.7,0){$\bullet$};
\node[red, draw, minimum size=4pt] (10) at (2.4,0){\tiny $\nearrow$};
\node (11) at (3,3){$\bullet$};
\node (12) at (2.4,2){$\bullet$};
\node (13) at (3.6,2){$\bullet$};
\draw [dotted ] (0) to (1);
%\draw (0) to (1);
\draw (2) to (1);
\draw (11) to (1);
\draw (2) to (3);
\draw (2) to (4);
\draw (4) to (5);
\draw (4) to (8);
\draw (5) to (6);
\draw (5) to (7);
\draw (8) to (9);
\draw (8) to (10);
\draw (11) to (12);
\draw (11) to (13);
\end{tikzpicture}
\hspace{-1mm}$\rightsquigarrow$\hspace{-1mm}%
\begin{tikzpicture}[baseline=(2),inner sep=1pt, circle, minimum size=0pt, xscale=.9, scale =0.7]\node (p) at (0,0){};
\node(0) at  (1.4,4.6) {};
\node (1) at (2,4){$\bullet$};
\node (2) at (1,3){$\bullet$};
\node (3) at (.4,2){$\bullet$};
\node (4) at (1.6,2){$\bullet$};
\node (5) at (1.1,1){$\bullet$};
\node (6) at (.7,0){$\bullet$};
\node (7) at (1.4,0){$\bullet$};
\node (8) at (2.1,1){$\bullet$};
\node (9) at (1.7,0){$\bullet$};
\node[draw, minimum size=4pt] (10) at (2.4,0){$\bullet$};
\node (11) at (3,3){$\bullet$};
\node (12) at (2.4,2){$\bullet$};
\node (13) at (3.6,2){$\bullet$};
\draw [dotted ] (0) to (1);
%\draw (0) to (1);
\draw (2) to (1);
\draw (11) to (1);
\draw (2) to (3);
\draw (2) to (4);
\draw (4) to (5);
\draw (4) to (8);
\draw (5) to (6);
\draw (5) to (7);
\draw (8) to (9);
\draw (8) to (10);
\draw (11) to (12);
\draw (11) to (13);
\end{tikzpicture}\\[2mm]
\begin{tikzpicture}[baseline=(2),inner sep=1pt, circle, minimum size=0pt, xscale=.9, scale =0.7]\node (p) at (0,0){};
\node(0) at  (1.4,4.6) {};
\node (1) at (2,4){$\bullet$};
\node (2) at (1,3){$\bullet$};
\node (3) at (.4,2){$\bullet$};
\node (4) at (1.6,2){$\bullet$};
\node (5) at (1.1,1){$\bullet$};
\node[red, draw, minimum size=4pt] (6) at (.7,0){\tiny $\nearrow$};
\node (7) at (1.4,0){$\bullet$};
\node (8) at (2.1,1){$\bullet$};
\node (9) at (1.7,0){$\bullet$};
\node (10) at (2.4,0){$\bullet$};
\node (11) at (3,3){$\bullet$};
\node (12) at (2.4,2){$\bullet$};
\node (13) at (3.6,2){$\bullet$};
\draw [dotted ] (0) to (1);
%\draw (0) to (1);
\draw (2) to (1);
\draw (11) to (1);
\draw (2) to (3);
\draw (2) to (4);
\draw (4) to (5);
\draw (4) to (8);
\draw (5) to (6);
\draw (5) to (7);
\draw (8) to (9);
\draw (8) to (10);
\draw (11) to (12);
\draw (11) to (13);
\end{tikzpicture}
\hspace{-1mm}$\rightsquigarrow$\hspace{-1mm}%
\begin{tikzpicture}[baseline=(2),inner sep=1pt, circle, minimum size=0pt, xscale=.9, scale =0.7]\node (p) at (0,0){};
\node(0) at  (1.4,4.6) {};
\node (1) at (2,4){$\bullet$};
\node (2) at (1,3){$\bullet$};
\node (3) at (.4,2){$\bullet$};
\node[draw, minimum size=4pt] (4) at (1.6,2){$\bullet$};
\node (5) at (1.1,1){$\bullet$};
\node (6) at (.7,0){$\bullet$};
\node (7) at (1.4,0){$\bullet$};
\node (8) at (2.1,1){$\bullet$};
\node (9) at (1.7,0){$\bullet$};
\node (10) at (2.4,0){$\bullet$};
\node (11) at (3,3){$\bullet$};
\node (12) at (2.4,2){$\bullet$};
\node (13) at (3.6,2){$\bullet$};
\draw [dotted ] (0) to (1);
%\draw (0) to (1);
\draw (2) to (1);
\draw (11) to (1);
\draw (2) to (3);
\draw (2) to (4);
\draw (4) to (5);
\draw (4) to (8);
\draw (5) to (6);
\draw (5) to (7);
\draw (8) to (9);
\draw (8) to (10);
\draw (11) to (12);
\draw (11) to (13);
\end{tikzpicture}

\end{center}
\caption{Trees of $\cC^{\star}$ and their associated pointed tree of $\cC+\cC^{\bullet}$}
\label{fig:repointcatalan}
\end{figure}

%% file: figurepointmotzkin.tex
\begin{figure}
\begin{center}
\begin{tikzpicture}[baseline=(2),inner sep=1pt, circle, minimum size=0pt, xscale=.9, scale =0.7]\node (p) at (0,0){};
\node(0) at  (1.4,4.6) {};
\node (1) at (2,4){$\bullet$};
\node (2) at (1,3){$\bullet$};
\node[green, draw, minimum size=4pt] (3) at (1,2){$\downarrow$};
\node (4) at (1,1){$\bullet$};
\node (5) at (1,0){$\bullet$};
\node (6) at (0,-1){$\bullet$};
\node (7) at (2,-1){$\bullet$};
\node (8) at (3,3){$\bullet$};
\draw [dotted ] (0) to (1);
%\draw (0) to (1);
\draw (2) to (1);
\draw (8) to (1);
\draw (2) to (3);
\draw (5) to (4);
\draw (4) to (3);
\draw (6) to (5);
\draw (5) to (7);
\end{tikzpicture}
\hspace{-1mm}$\rightsquigarrow$\hspace{-5mm}%
\begin{tikzpicture}[baseline=(2),inner sep=1pt, circle, minimum size=0pt, xscale=.9, scale =0.7]\node (p) at (0,0){};
\node(0) at  (1.4,4.6) {};
\node (1) at (2,4){$\bullet$};
\node (2) at (1,3){$\bullet$};
\node (3) at (1,2){$\bullet$};
\node[draw, minimum size=4pt] (4) at (1,1){$\bullet$};
\node (5) at (1,0){$\bullet$};
\node (6) at (0,-1){$\bullet$};
\node (7) at (2,-1){$\bullet$};
\node (8) at (3,3){$\bullet$};
\draw [dotted ] (0) to (1);
%\draw (0) to (1);
\draw (2) to (1);
\draw (8) to (1);
\draw (4) to (3);
\draw (2) to (3);
\draw (5) to (4);
\draw (6) to (5);
\draw (5) to (7);
\end{tikzpicture}
\hspace{.5cm}
\begin{tikzpicture}[baseline=(2),inner sep=1pt, circle, minimum size=0pt, xscale=.9, scale =0.7]\node (p) at (0,0){};
\node(0) at  (1.4,4.6) {};
\node (1) at (2,4){$\bullet$};
\node (2) at (1,3){$\bullet$};
\node (3) at (1,2){$\bullet$};
\node (4) at (1,1){$\bullet$};
\node (5) at (1,0){$\bullet$};
\node[red, draw, minimum size=4pt] (6) at (0,-1){\tiny $\nearrow$};
\node (7) at (2,-1){$\bullet$};
\node (8) at (3,3){$\bullet$};
\draw [dotted ] (0) to (1);
%\draw (0) to (1);
\draw (2) to (1);
\draw (8) to (1);
\draw (2) to (3);
\draw (5) to (4);
\draw (6) to (5);
\draw (5) to (7);
\draw (4) to (3);
\end{tikzpicture}
\hspace{-1mm}$\rightsquigarrow$\hspace{-5mm}%
\begin{tikzpicture}[baseline=(2),inner sep=1pt, circle, minimum size=0pt, xscale=.9, scale =0.7]\node (p) at (0,0){};
\node(0) at  (1.4,4.6) {};
\node[draw, minimum size=4pt] (1) at (2,4){$\bullet$};
\node (2) at (1,3){$\bullet$};
\node (3) at (1,2){$\bullet$};
\node (4) at (1,1){$\bullet$};
\node (5) at (1,0){$\bullet$};
\node (6) at (0,-1){$\bullet$};
\node (7) at (2,-1){$\bullet$};
\node (8) at (3,3){$\bullet$};
\draw [dotted ] (0) to (1);
%\draw (0) to (1);
\draw (2) to (1);
\draw (8) to (1);
\draw (2) to (3);
\draw (5) to (4);
\draw (6) to (5);
\draw (5) to (7);
\draw (4) to (3);
\end{tikzpicture}\\[3mm]
\begin{tikzpicture}[baseline=(2),inner sep=1pt, circle, minimum size=0pt, xscale=.9, scale =0.7]\node (p) at (0,0){};
\node(0) at  (1.4,4.6) {};
\node (1) at (2,4){$\bullet$};
\node (2) at (1,3){$\bullet$};
\node (3) at (1,2){$\bullet$};
\node (4) at (1,1){$\bullet$};
\node (5) at (1,0){$\bullet$};
\node (6) at (0,-1){$\bullet$};
\node[blue, draw, minimum size=4pt] (7) at (2,-1){\tiny $\nwarrow$};
\node (8) at (3,3){$\bullet$};
\draw [dotted ] (0) to (1);
%\draw (0) to (1);
\draw (2) to (1);
\draw (8) to (1);
\draw (2) to (3);
\draw (5) to (4);
\draw (6) to (5);
\draw (5) to (7);
\draw (4) to (3);
\end{tikzpicture}
\hspace{-1mm}$\rightsquigarrow$\hspace{-5mm}%
\begin{tikzpicture}[baseline=(2),inner sep=1pt, circle, minimum size=0pt, xscale=.9, scale =0.7]\node (p) at (0,0){};
\node(0) at  (1.4,4.6) {};
\node (1) at (2,4){$\bullet$};
\node[draw, minimum size=4pt] (2) at (1,3){$\bullet$};
\node (3) at (1,2){$\bullet$};
\node (4) at (1,1){$\bullet$};
\node (5) at (1,0){$\bullet$};
\node (6) at (0,-1){$\bullet$};
\node (7) at (2,-1){$\bullet$};
\node (8) at (3,3){$\bullet$};
\draw [dotted ] (0) to (1);
%\draw (0) to (1);
\draw (2) to (1);
\draw (8) to (1);
\draw (2) to (3);
\draw (5) to (4);
\draw (4) to (3);
\draw (6) to (5);
\draw (5) to (7);
\end{tikzpicture}

\end{center}
\caption{Trees of $\cM^{\star}$ and their associated pointed tree of $\cM+\cM^{\bullet}$}
\label{fig:repointcmotzkin}
\end{figure}

%% file: exemplecatalanmieux.tex
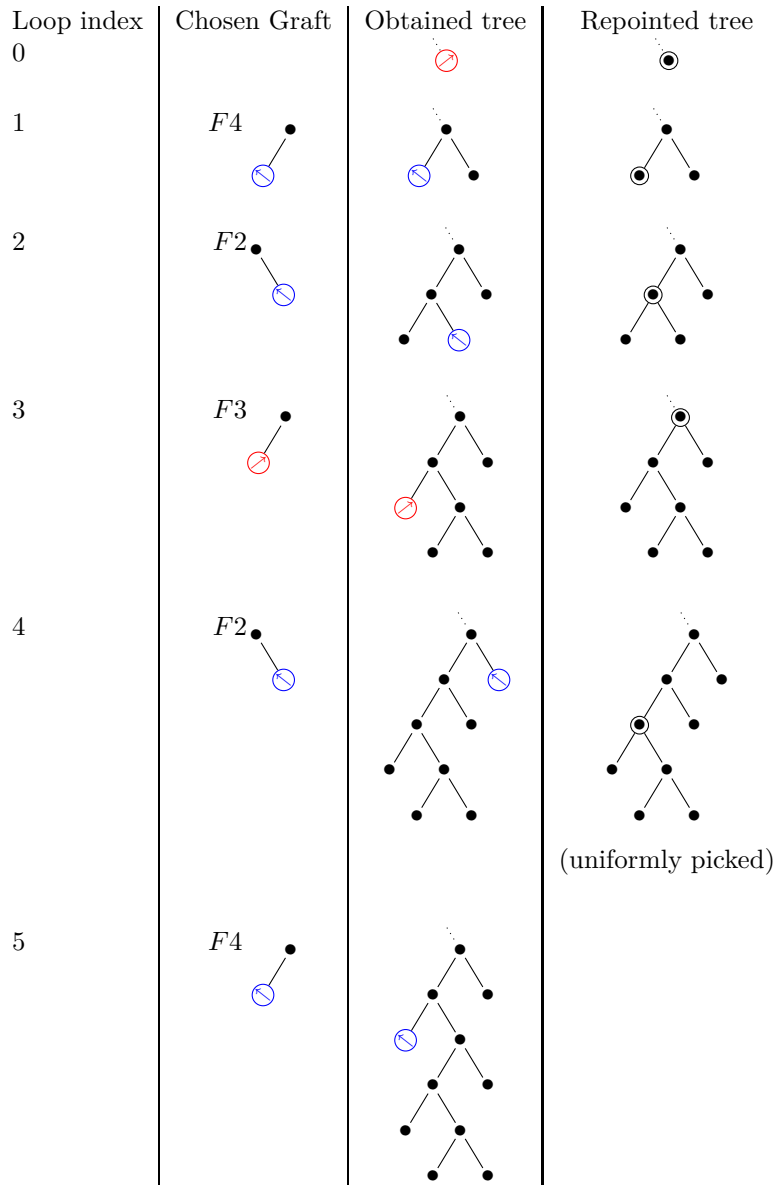
\begin{figure}
\begin{center}
\begin{tabular}{l|c|c|c}
Loop index&Chosen Graft&Obtained tree&Repointed tree\\
0&&
\begin{tikzpicture}[baseline=(1),inner sep=0pt, circle, minimum size=0pt, scale=.6, xscale=0.6]
\node (r) at (.5,-.5) {};
\node[red, draw, minimum size=4pt] (1) at (1,-1) {\tiny $\nearrow$};
\draw[dotted] (r) to (1);
\end{tikzpicture}&
\begin{tikzpicture}[baseline=(1),inner sep=0pt, circle, minimum size=0pt, scale=.6, xscale=0.6]
\node (r) at (.5,-.5) {};
\node[draw, minimum size=5pt] (1) at (1,-1) {$\bullet$};
\draw[dotted] (r) to (1);
\end{tikzpicture}\\[.5cm]
1&$F4$ \begin{tikzpicture}[baseline=(1),inner sep=0pt, circle, minimum size=0pt, scale=.6, xscale=0.6]
\node (1) at (1,-1) {$\bullet$};
\node (chut) at (.5,-1.5) {};
\node[blue, draw, minimum size=4pt] (2) at (0,-2) {\tiny $\nwarrow$};
\draw (1) to (2);
\end{tikzpicture}&
\begin{tikzpicture}[baseline=(1),inner sep=0pt, circle, minimum size=0pt, scale=.6, xscale=0.6]
\node (r) at (.5,-.5) {};
\node (1) at (1,-1) {$\bullet$};
\node[blue, draw, minimum size=4pt] (2) at (0,-2) {\tiny $\nwarrow$};
\node (3) at (2,-2) {$\bullet$};
\draw[dotted] (r) to (1);
\draw (1) to (2);
\draw (1) to (3);
\end{tikzpicture}&
\begin{tikzpicture}[baseline=(1),inner sep=0pt, circle, minimum size=0pt, scale=.6, xscale=0.6]
\node (r) at (.5,-.5) {};
\node (1) at (1,-1) {$\bullet$};
\node[draw, minimum size=4pt] (2) at (0,-2) {$\bullet$};
\node (3) at (2,-2) {$\bullet$};
\draw[dotted] (r) to (1);
\draw (1) to (2);
\draw (1) to (3);
\end{tikzpicture}\\[1.15cm]
2&$F2$\begin{tikzpicture}[baseline=(1),inner sep=0pt, circle, minimum size=0pt, scale=.6, xscale=0.6]
\node (1) at (1,-1) {$\bullet$};
\node (chut) at (1.5,-1.5) {};
\node[blue, draw, minimum size=4pt] (2) at (2,-2) {\tiny $\nwarrow$};
\draw (1) to (2);
\end{tikzpicture}&
\begin{tikzpicture}[baseline=(1),inner sep=0pt, circle, minimum size=0pt, scale=.6, xscale=0.6]
\node (r) at (.5,-.5) {};
\node (1) at (1,-1) {$\bullet$};
\node (2) at (0,-2) {$\bullet$};
\node (3) at (2,-2) {$\bullet$};
\node (4) at (-1,-3) {$\bullet$};
\node[blue, draw, minimum size=4pt] (5) at (1,-3) {\tiny $\nwarrow$};
\draw[dotted] (r) to (1);
\draw (1) to (2);
\draw (1) to (3);
\draw (4) to (2);
\draw (5) to (2);
\end{tikzpicture}&
\begin{tikzpicture}[baseline=(1),inner sep=0pt, circle, minimum size=0pt, scale=.6, xscale=0.6]
\node (r) at (.5,-.5) {};
\node (1) at (1,-1) {$\bullet$};
\node[draw, minimum size=4pt] (2) at (0,-2) {$\bullet$};
\node (3) at (2,-2) {$\bullet$};
\node (4) at (-1,-3) {$\bullet$};
\node (5) at (1,-3) {$\bullet$};
\draw[dotted] (r) to (1);
\draw (1) to (2);
\draw (1) to (3);
\draw (4) to (2);
\draw (5) to (2);
\end{tikzpicture}\\[1.8cm]
3&$F3$\begin{tikzpicture}[baseline=(1),inner sep=0pt, circle, minimum size=0pt, scale=.6, xscale=0.6]
\node (1) at (1,-1) {$\bullet$};
\node (chut) at (.5,-1.5) {};
\node[red, draw, minimum size=4pt] (2) at (0,-2) {\tiny $\nearrow$};
\draw (1) to (2);
\end{tikzpicture}&
\begin{tikzpicture}[baseline=(1),inner sep=0pt, circle, minimum size=0pt, scale=.6, xscale=0.6]
\node (r) at (.5,-.5) {};
\node (1) at (1,-1) {$\bullet$};
\node (2) at (0,-2) {$\bullet$};
\node (3) at (2,-2) {$\bullet$};
\node[red, draw, minimum size=4pt] (4) at (-1,-3) {\tiny $\nearrow$};
\node (5) at (1,-3) {$\bullet$};
\node (6) at (0,-4) {$\bullet$};
\node (7) at (2,-4) {$\bullet$};
\draw[dotted] (r) to (1);
\draw (1) to (2);
\draw (1) to (3);
\draw (4) to (2);
\draw (5) to (2);
\draw (5) to (6);
\draw (5) to (7);
\end{tikzpicture}&
\begin{tikzpicture}[baseline=(1),inner sep=0pt, circle, minimum size=0pt, scale=.6, xscale=0.6]
\node (r) at (.5,-.5) {};
\node[draw, minimum size=4pt] (1) at (1,-1) {$\bullet$};
\node (2) at (0,-2) {$\bullet$};
\node (3) at (2,-2) {$\bullet$};
\node (4) at (-1,-3) {$\bullet$};
\node (5) at (1,-3) {$\bullet$};
\node (6) at (0,-4) {$\bullet$};
\node (7) at (2,-4) {$\bullet$};
\draw[dotted] (r) to (1);
\draw (1) to (2);
\draw (1) to (3);
\draw (4) to (2);
\draw (5) to (2);
\draw (5) to (6);
\draw (5) to (7);
\end{tikzpicture}\\[2.45cm]
4&$F2$\begin{tikzpicture}[baseline=(1),inner sep=0pt, circle, minimum size=0pt, scale=.6, xscale=0.6]
\node (1) at (1,-1) {$\bullet$};
\node (chut) at (1.5,-1.5) {};
\node[blue, draw, minimum size=4pt] (2) at (2,-2) {\tiny $\nwarrow$};
\draw (1) to (2);
\end{tikzpicture}&
\begin{tikzpicture}[baseline=(1),inner sep=0pt, circle, minimum size=0pt, scale=.6, xscale=0.6]
\node (r) at (.5,-.5) {};
\node (1) at (1,-1) {$\bullet$};
\node (2) at (0,-2) {$\bullet$};
\node[blue, draw, minimum size=4pt] (3) at (2,-2) {\tiny $\nwarrow$};
\node (4) at (-1,-3) {$\bullet$};
\node (5) at (1,-3) {$\bullet$};
\node (6) at (0,-4) {$\bullet$};
\node (7) at (-2,-4) {$\bullet$};
\node (8) at (-1,-5) {$\bullet$};
\node (9) at (1,-5) {$\bullet$};
\draw[dotted] (r) to (1);
\draw (1) to (2);
\draw (1) to (3);
\draw (4) to (2);
\draw (5) to (2);
\draw (4) to (6);
\draw (4) to (7);
\draw (8) to (6);
\draw (9) to (6);
\end{tikzpicture}&
\begin{tikzpicture}[baseline=(1),inner sep=0pt, circle, minimum size=0pt, scale=.6, xscale=0.6]
\node (r) at (.5,-.5) {};
\node (1) at (1,-1) {$\bullet$};
\node (2) at (0,-2) {$\bullet$};
\node (3) at (2,-2) {$\bullet$};
\node[draw, minimum size=4pt] (4) at (-1,-3) {$\bullet$};
\node (5) at (1,-3) {$\bullet$};
\node (6) at (0,-4) {$\bullet$};
\node (7) at (-2,-4) {$\bullet$};
\node (8) at (-1,-5) {$\bullet$};
\node (9) at (1,-5) {$\bullet$};
\draw[dotted] (r) to (1);
\draw (1) to (2);
\draw (1) to (3);
\draw (4) to (2);
\draw (5) to (2);
\draw (4) to (6);
\draw (4) to (7);
\draw (8) to (6);
\draw (9) to (6);
\node at (0,-3) {\phantom{(uniformly picked)}};
\useasboundingbox (current bounding box);
\node at (0,-6) {(uniformly picked)};
\end{tikzpicture}\\[3.75cm]
5&$F4$ \begin{tikzpicture}[baseline=(1),inner sep=0pt, circle, minimum size=0pt, scale=.6, xscale=0.6]
\node (1) at (1,-1) {$\bullet$};
\node (chut) at (.5,-1.5) {};
\node[blue, draw, minimum size=4pt] (2) at (0,-2) {\tiny $\nwarrow$};
\draw (1) to (2);
\end{tikzpicture}&
\begin{tikzpicture}[baseline=(1),inner sep=0pt, circle, minimum size=0pt, scale=.6, xscale=0.6]
\node (r) at (.5,-.5) {};
\node (1) at (1,-1) {$\bullet$};
\node (2) at (0,-2) {$\bullet$};
\node (3) at (2,-2) {$\bullet$};
\node[blue, draw, minimum size=4pt] (4) at (-1,-3) {\tiny $\nwarrow$};
\node (5) at (1,-3) {$\bullet$};
\node (6) at (0,-4) {$\bullet$};
\node (7) at (2,-4) {$\bullet$};
\node (8) at (-1,-5) {$\bullet$};
\node (9) at (1,-5) {$\bullet$};
\node (10) at (0,-6) {$\bullet$};
\node (11) at (2,-6) {$\bullet$};
\draw[dotted] (r) to (1);
\draw (1) to (2);
\draw (1) to (3);
\draw (4) to (2);
\draw (5) to (2);
\draw (5) to (6);
\draw (5) to (7);
\draw (8) to (6);
\draw (9) to (6);
\draw (9) to (10);
\draw (9) to (11);
\end{tikzpicture}&
\end{tabular}
\end{center}
\caption{An example of a run of Algorithm~\ref{algo:efficient} producing a
tree with $2\times5+1 = 11$ nodes.}
\label{fig:exrun}
\end{figure}

%% file: articleAxel.bbl
\begin{thebibliography}{10}

\bibitem{Alonso1994529}
L.~Alonso.
\newblock Uniform generation of a motzkin word.
\newblock {\em Theoretical Computer Science}, 134(2):529 -- 536, 1994.

\bibitem{Alonso96alinear}
L.~Alonso, J.L. Rémy, and R.~Schott.
\newblock A linear time algorithm for the generation of trees, 1996.

\bibitem{Barcucci19953}
E.~Barcucci, R.~Pinzani, and R.~Sprugnoli.
\newblock The random generation of underdiagonal walks.
\newblock {\em Discrete Mathematics}, 139(1–3):3 -- 18, 1995.

\bibitem{bodinicissedavid}
O.~Bodini, M.~Cisse, and J.~David.
\newblock Private communication, 2013.

\bibitem{bodiniponty}
O.~Bodini and Y.~Ponty.
\newblock Multi-dimensional boltzmann sampling of languages.
\newblock In {\em DMTCS Proceedings}, pages 49--64, Vienne, Autriche, 2010.
\newblock 12pp.

\bibitem{conf/wsc/Duchon11}
Ph. Duchon.
\newblock Random generation of combinatorial structures: Boltzmann samplers and
  beyond.
\newblock In S.~Jain, Roy R.~Creasey Jr., Jan Himmelspach, K.~Preston White,
  and Michael~C. Fu, editors, {\em Winter Simulation Conference}, pages
  120--132. WSC, 2011.

\bibitem{bolt1}
Ph. Duchon, Ph. Flajolet, G.~Louchard, and G.~Schaeffer.
\newblock Boltzmann samplers for the random generation of combinatorial
  structures.
\newblock {\em Combinatorics, Probability and Computing}, 13:577--625, 2004.

\bibitem{livre}
P.~Flajolet and R.~Sedgewick.
\newblock {\em Analytic combinatorics}.
\newblock Cambridge University Press, Cambridge, 2009.

\bibitem{FlPeSo11}
Ph. Flajolet, M.~Pelletier, and M.~Soria.
\newblock On {Buffon} machines and numbers.
\newblock In Dana Randall, editor, {\em Proceedings of the Twenty-Second Annual
  ACM-SIAM Symposium on Discrete Algorithms, SODA 2011}, pages 172--183. SIAM,
  2011.

\bibitem{FZV94}
Philippe Flajolet, Paul Zimmermann, and Bernard Van~Cutsem.
\newblock A calculus of random generation of labelled combinatorial structures.
\newblock {\em Theoret. Comput. Sci.}, 132:1--35, 1994.

\bibitem{Remy85}
J.L. Remy.
\newblock Un proc{\'e}d{\'e} it{\'e}ratif de d{\'e}nombrement d'arbres binaires
  et son application a leur g{\'e}n{\'e}ration al{\'e}atoire.
\newblock {\em ITA}, 19(2):179--195, 1985.

\bibitem{SaZi94}
B.~Salvy and P.~Zimmermann.
\newblock Gfun: a {M}aple package for the manipulation of generating and
  holonomic functions in one variable.
\newblock {\em ACM Transactions on Mathematical Software}, 20(2):163--177,
  1994.

\bibitem{shannon1948mathematical}
Claude~E. Shannon.
\newblock A mathematical theory of communication.
\newblock {\em The Bell System Technical Journal}, 27:379--423, 623--656, July,
  October 1948.

\end{thebibliography}
